\documentclass[twocolumn]{IEEEtran}
\usepackage{amsmath}
\usepackage{amssymb}
\usepackage{amsfonts}
\usepackage{epsfig}
\usepackage{cite}
\usepackage{subfigure}
\usepackage[all,2cell]{xy} \UseAllTwocells \SilentMatrices
\usepackage{framed}
\usepackage{soul} %for highlight \hl{text}
\usepackage{color}
\usepackage{datetime}

\usepackage[bookmarks=false,pdfstartview=FitH]{hyperref}
\DeclareGraphicsExtensions{.pdf}%
\usepackage{command}

\newtheorem{defn}{\bf Definition}
\newtheorem{thm}{\bf Theorem}
\newtheorem{prop}[thm]{\bf Proposition}

\newcommand{\tmR}{\tm{R}}
\newcommand{\tmF}{\tm{F}}
\usepackage{algorithm,algorithmic}

\renewcommand{\alert}[1]{#1}

\title{Achieving Large Sum Rate and Good Fairness in MISO Broadcast Communication}
\author{Ji-You Huang and Hsiao-feng Francis  Lu
\thanks{H. F. Lu is with the Department of Electrical and Computer Engineering, National Chiao Tung University, Hsinchu, Taiwan (e-mail:francis@mail.nctu.edu.tw). The research of H. F. Lu was funded in part by Taiwan Ministry of Science and Technology under Grants MOST 106-2221-E-009 -024 -MY3 and MOST 107-2918-I-009-011. }
}
\begin{document}
\bibliographystyle{ieeetran}
\allowdisplaybreaks
%\newpage
%\thispagestyle{empty}
%\tableofcontents
%\newpage
%\setcounter{page}{1}
\maketitle

\begin{abstract}
A tradeoff between sum rate and fairness for MISO broadcast communication employing dirty paper coding or zero-forcing dirty paper coding at physical layer is investigated in this paper. The tradeoff is based on a new design objective termed "tri-stage" approach as well as a new $\ell_1$-based fairness measure that is much more robust than the well-known Jain's index for comparing fairness levels achieved by various design objectives at a much finer resolution in high SNR regime. The newly proposed tri-stage design also introduces a new concept of statistical power allocation that randomly allocates powers to users based on an optimal probability distribution derived from the tradeoff between sum rate and fairness. Simulation results show that the proposed approach can simultaneously achieve a larger sum rate and better fairness than the reputable proportional fairness criterion. A performance upper bound is also given in the paper to show that the excellent performance of the proposed approach at moderate and high SNR regimes as well as some potential for further improvement in low SNR regime. 
\end{abstract}


\section{Introduction} \label{sec:intro}

In the downlink of wireless communication, a multi-antenna transmitter could send information simultaneously to multiple single-antenna users. Such communication channel is commonly referred to as the {\em multiple-input single-output} (MISO) {\em broadcast channel} (BC) in the literatures \cite{IT_CaireS03}. It is the dual of SIMO (single-input multiple-output) multiple-access channel (MAC) \cite{VishwanathJG03}, where multiple single-antenna users send information simultaneously to a common multi-antenna receiver. The MISO BC has appeared widely not only in traditional mobile communications, but also in the latest Internet of Things (IoT) and device-to-device (D2D) communication systems \cite{conf_HuKHR16}. 

When broadcasting information to all users, the transmitter could apply the well-known dirty-paper-coding (DPC) scheme \cite{Costa83,ErezB05} for encoding  messages, provided that 1) a certain ordering of users has been established a priori, 2) each user has a perfect knowledge of the channel state information (CSI) of his/her incoming channels, and 3) the transmitter has a complete knowledge of the CSI of all users. The second requirement is commonly referred to  as the CSIR --- CSI at receiver --- and can be achieved through channel estimation; the last requirement is coined as the CSIT --- CSI at transmitter --- in the literature and can be realized by using a feedback channel from the users to the transmitter. These requirements can be easily achieved, as the mechanisms for channel estimation and feedback already exist in modern communication systems \cite{IEEE802.11mac}. 
%For situations such as only statistical CSI \cite{CaoSNLT17} or delayed CSI \cite{LuoR15} is available at the transmitter, DPC cannot be employed. 
Armed with the ordering and complete knowledge of CSI, the DPC strategy successively encodes each user's message taking into account the noncausal knowledge of interference signals caused by preceding users. The scheme then converts the broadcast channel into a special kind of the {\em Gelfand-Pinsker channel}  \cite{Gelfand1980} with states non-causally known at the transmitter. The DPC strategy not only achieves the capacity of the Gelfand-Pinsker channel \cite{Costa83, book_Kramer08}, but also turns out to be optimal for MIMO BC. To elaborate, Caire and Shamai \cite{IT_CaireS03} investigated the capacity region of the two-user BC when the base station has  arbitrary number of transmit antennas  and each user has only single receive antenna. They showed through direct calculation that the DPC is optimal in terms of achieving the sum capacity of the two-user MISO BC. Weingarten \etal \cite{conf_WSS04,IT_Wei06} studied the capacity region of Gaussian MIMO BC  based on the notion of an {\em enhanced broadcast channel} under a wide range of input constraints, including the total power and per-antenna constraints. They showed that the capacity region coincides with the DPC rate region. An alternative proof for the capacity region of degraded Gaussian MIMO BC without using the notion of enhanced channel can be found in \cite{EkremU12}. 

Many iterative algorithms \cite{JindalRVJG05,Yu06a} have been proposed to find the optimal coding (beamforming) vector associated to each user in the DPC strategy, aiming to maximize the overall sum rate. These algorithms are based on the duality between the MAC and BC \cite{SchubertB04} and suffer from a relatively high computational complexity. In \cite{IT_CaireS03} Caire and Shamai proposed a suboptimal transmission strategy, termed {\em zero-forcing DPC} (ZFDPC), that combines both the advantages of zero-forcing beamformer and DPC for MISO BC, when the number of transmit antennas at base station exceeds or equals the number of single-antenna users. The ZFDPC eventually decomposes the MISO BC into a group of parallel interference-free channels and simplifies the problem of finding optimal coding vectors, but is at a cost of certain capacity loss.  Since then, the idea of combining zero forcing and DPC has been applied to many other communication problems. For instance, Dabbagh and Love \cite{DabbaghL07} extended the work in \cite{IT_CaireS03} and proposed a successive ZFDPC encoder for the MIMO broadcast channel, i.e. when the users are equipped with multiple receive antennas. Mohammed and Larsson \cite{MohammedL16} proposed a user-group based ZFDPC precoder by splitting users into disjoint groups. Hu and Rusek \cite{HuR17} considered a generalized zero-forcing beamforming that is not only orthogonal to the channel vectors of the succeeding users but also to part of those of the preceding users, thereby yielding a generalized ZFDPC strategy, where the DPC encoder only has to take into account the multiuser interference caused by a small constant number of immediately preceding users. 

All the above works share a common objective, namely, maximizing the sum of transmission rates of users \cite{book_Kramer08}, and care less whether the scheme is equally beneficial to the individuals. In other words, there can be two (contradicting) objectives for designing communication schemes for BC, one from the transmitter's viewpoint, i.e. sum rate maximization (as all the above schemes do), and the other from the viewpoint of each user, i.e. fairness maximization. Generally speaking, the former objective can be mathematically and quantitively expressed using the formula of achievable rates of the communication scheme used, say DPC, ZFDPC, ZF beamforming, etc., but the notion of fairness is unfortunately much harder to be quantified. Several conceptual, philosophical and qualitative definitions of fairness, such as proportional fairness  \cite{Kelly98}, harmonic mean fairness \cite{LuoZ08}, max-min fairness  \cite{MoW00}, etc. have been proposed in the literature, each holding a different opinion regarding how it means to be fair. There are also some quantitive measures for fairness in the literatures. Plausible fairness measures are generally required to satisfy axioms such as continuity, homogeneity, asymptotic saturation, and monotonicity \cite{conf_LanKCS10}. Examples of such fairness measures are the entropy-based index \cite{conf_UchidaK09},  Jain’s fairness index \cite{Jain84}, $\alpha$-fairness from networking research community \cite{Kelly98,MoW00}, and a much more complicated construction \cite{conf_LanKCS10}  that includes many existing measures as special cases. In particular, $\alpha$-fairness measure can be used to justify some of the aforementioned  qualitative approaches for fairness  by varying the parameter $\alpha$. 
For instance, setting $\alpha=0$ yields the aim of sum rate maximization \cite{conf_LanKCS10}. Setting $\alpha=1$ gives the proportional fairness criterion, and the case of $\alpha=\infty$ corresponds to the max-min fairness. 
Studies of tradeoffs between sum rate and fairness also appear in literatures.  In \cite{SediqGSY13}, Sediq \etal investigated such tradeoff at network level based on Jain's index and $\alpha$-fair utility \cite{MoW00}. In particular, the transmission rates in \cite{SediqGSY13} were replaced by the numbers of resource blocks allocated to each user in the downlink of wireless networks using OFDM, thereby yielding an orthogonal communication. Such communication scheme is extremely suboptimal in terms of maximal achievable rates from the viewpoint of multi-terminal information theory \cite{Cover,  book_Kramer08, IT_CaireS03, IT_Wei06, EkremU12}, since orthogonal schemes have sum degrees of freedom always equal to one, regardless of the increase of transmit antennas and users. It then leaves a significant room for improving the tradeoff between rates and fairness. 

In this paper, we will investigate the tradeoff between the sum rate and fairness for MISO BC at physical layer by employing communication schemes such as DPC or ZFDPC to guarantee the close-to-capacity performance. In addition, we will aim to provide a systematic design that can offer not only a good sum rate but also a reasonable fairness to all users at the same time. To this end, we will first review in Section \ref{sec:channel} the system model of MISO BC as well as the DPC and ZFDPC strategies. In Section \ref{sec:qual} we will discuss several commonly used design objectives derived from the qualitative notions of fairness.  Quantitive fairness measures will be discussed in Section \ref{sec:fair}. In particular, it will be seen that the Jain's index, though widely accepted as a fairness measure, behaves less sensitive to the fluctuations of transmission rates in high SNR regime. An alternative measure based on $\ell_1$-norm will be proposed in Section \ref{sec:fair} for replacement. The new measure is easily computable and satisfies almost all axioms listed in \cite{conf_LanKCS10}. In Section \ref{sec:tri} we will present the new design objective, termed {\em tri-stage}, which takes into account both the qualitative and quantitive aspects of fairness. As the name suggests, the proposed approach consists of three stages, where the first two stages aim to obtain a tradeoff between sum rate and fairness as well as a byproduct which will be discussed next. Note that the tradeoff is just a function relating sum rate to fairness and says nothing about which pair of sum rate and fairness should be chosen for operation. Choosing the operating pair can be philosophically hard. For instance, it is arguable to allege that scheme A having a sum rate of 10 bits per channel use and 90 per cent fairness is better than scheme B having a sum rate of 11 bits per channel use and 80 per cent fairness, and vice versa. One might assert that schemes A and B are equally good as they are both {\em Pareto optimal points} on the sum rate-fairness tradeoff curve from the viewpoint of operational research. Yet, it would be universally agreed that another Pareto optimal, equally good scheme C having a sum rate of 12 bits per channel use and 10 per cent fairness should be totally unacceptable. Thus, we will turn to the qualitative notion of fairness to decide the operating pair in the third stage of the proposed design. Achieving the sum rate and fairness of the chosen operating point makes use of the byproduct obtained in the first two stages: a new concept of statistical power allocation, which is in sharp contrast to the fixed, deterministic method used in all existing wireless/wired communication systems. The new scheme randomly --- based on an optimal probability distribution derived from the tradeoff --- allocates powers to users, thereby offering not only a larger sum rate but also a better fairness than the existing designs. Several simulation results will be provided in Section \ref{sec:sim} to justify the excellent performance of the proposed approach. Concluding remarks are given in Section \ref{sec:con}. 

The following notations have been used in this paper. Underlined lowercase letter $\underline{x}$ represents a vector, and uppercase letter $A$ denotes a matrix of certain size. $A^\dag$ (resp. $A^\top$) denotes the Hermitian transpose (resp. transpose) of matrix $A$, and $\norm{A}_p$ denotes its $\ell_p$ norm for some $p \geq 1$. $I_n$ is the $(n \times n)$ identity matrix. $\< \_{a}, \_{b}\>$ is the usual Euclidean inner product for $\_{a}, \_{b} \in \R^n$.  Matrix inequalities such as $\succeq$, $\preceq$, $\succ$ and $\prec$, are the partial orderings of positive semi-definite matrices \cite[Section 7.7]{Horn}. We say $\underline{x} \sim \CN{\underline{m}}{K}$ when $\underline{x}$ is a circularly symmetric complex Gaussian random vector with mean $\underline{m}$ and covariance matrix $K$.

\section{The MISO Broadcast Channel} \label{sec:channel}

Consider a $K$-user MISO broadcast channel, where each user has only one antenna and the base-station has $N$ antennas. Assume that with a codebook ${\cal X} \subset \C^N$ the base-station transmits $\_{x} \in {\cal X}$  simultaneously to all users, subject to an average power constraint $\E \norm{\_{x}}_2^2 \leq P$. With an implicit ordering of users, the signal received by the $k$-th user is given by
\beq
y_k = \_{h}_k^\top \_{x} + w_k, \label{eq:1}
\eeq
where $\_{h}_k \in \C^N$ is the channel vector from the base-station to the $k$-th user, and $w_k \in \CN{0}{1}$ is the additive complex Gaussian noise associated with the channel. Assume further that the transmitted signal can be decomposed into
\beq
\_{x} = \_{x}_1 + \_{x}_2 + \cdots + \_{x}_K, \label{eq:2}
\eeq
where $\_{x}_k$ is the signal intended for the $k$-user and satisfies $\E \norm{\_{x}_k}_2^2 \leq p_k$. The signals $\_{x}_k$ are statistically uncorrelated, so we have $p_1 + \cdots + p_K \leq P$, i.e. the {\em total power constraint}. We will assume throughout the paper that the channel vectors $\{ \_{h}_k\}$ are known perfectly to the base-station, but each user, say user $k$, knows only the channel vector $\_{h}_k$ of his/her own incoming channel. Substituting \eqref{eq:2} into \eqref{eq:1} gives an alternative expression of the received signal of user $k$ \bea
y_k &=& \sum_{i < k} \_{h}_k^\top \_{x}_i + \_{h}_k^\top \_{x}_k + \sum_{i>k} \_{h}_k^\top \_{x}_i + w_k \no\\
&=& \_{h}_k^\top \_{s}_k  + \_{h}_k^\top \_{x}_k +  z_k, \label{eq:3}
\eea
where $\_{s}_k$ is the interference caused by preceding users and $z_k$ is the overall noise consisting of the signals of succeeding users and the Gaussian noise $w_k$. Specifically, we have 
\bea
\_{s}_k &:=& \sum_{i < k} \_{x}_i\ \\
z_k &:=& \sum_{i>k} \_{h}_k^\top \_{x}_i + w_k. 
\eea

\subsection{Coding Strategies} 
When the channel vectors $\{ \_{h}_k\}$ are all known completely to the transmitter, i.e. the full CSIT scenario, it is known that the capacity region of the MISO BC --- with respect to the specific ordering of users --- can be achieved by coding strategies such as DPC \cite{conf_WSS04,IT_Wei06},  which encodes the message of user $k$ taking into account the noncausal knowledge of interference signal $\_{s}_k$ caused by preceding users. Specifically, assume the signal vectors $\_{x}_k$'s  are all complex Gaussian, i.e. $\_{x}_k \sim \CN{\_{0}}{K_k}$ with $\tr(K_k) \leq p_k$; then covariance matrix for $\_{s}_k$ is 
\beq
S_{k} := \E \_{s}_k \_{s}_k^\dag = \sum_{i < k} K_i,
\eeq
as the $\_{x}_k$'s are uncorrelated by hypothesis. Focusing on the $k$-th user, the standard approach of DPC uses the following auxiliary random vector 
\beq
\_{v}_k := \_{x}_k + (\_{h}_k^\top \_{s}_k) \_{\beta}_k \label{eq:vk}
\eeq
to construct a random codebook for user $k$, where 
\beq
\_{\beta}_k := \frac{1}{1 + \_{h}_k^\dag \left( \sum_{i \geq k} K_i \right) \_{h}_k} K_k \_{h}_k^*
\eeq
is chosen such that $\_{v}_k - y_k \_{\beta}_k$ and $y_k$ are statistically uncorrelated. 

Regarding $\_{s}_k$ as the side-information known to the transmitter (i.e. the base-station encoder), but not to the receiver (i.e. the decoder of the $k$-th user), the maximal achievable rate of the $k$-th user equals that of the Gelfand-Pinsker channel and is given by 
\bea
R_{\tm{DPC},k} & = & I(\_{v}_k; y_k) - I(\_{v}_k; \_{s}_k) \no \\
&=& \log_2 \left( 1+\frac{ \_{h}_k^\dag K_k \_{h}_k}{1+\_{h}_k^\dag \left( \sum_{i > k} K_i \right) \_{h}_k} \right) \label{eq:9}
\eea
in bits per channel use. Thus, the maximal sum rate achieved by DPC in  MISO BC with respect to the specific ordering of users equals 
\beq
\begin{array}{lll}
R_{\tm{DPC,sum}} := & \tm{maximize} & \sum_k R_{\tm{DPC},k} \\
& \tm{subject to} & K_1, \ldots, K_K \succeq {\bf 0}\\
& & \tr( K_k) \leq p_k
\end{array} \label{eq:R_DPC_sum}
\eeq
which could be further increased by optimizing over all possible $K!$ orderings of users. 
Several iterative algorithms \cite{JindalRVJG05,Yu06a} have been proposed to tackle the complicated optimization problem \eqref{eq:R_DPC_sum}, but the required computational complexity is generally high. 

Caire and Shamai \cite{IT_CaireS03} proposed a much simpler but suboptimal transmission strategy, termed {\em zero-forcing DPC} (ZFDPC), that combines both advantages of zero-forcing beamformer and DPC for MISO BC when $K \leq N$ (the situation when the base station has more transmit antennas than the number of single-antenna users).  Given a pre-determined ordering of users, the ZFDPC encoder uses zero-forcing beamformer to ensure that the coding (beamforming) vector of each user is orthogonal to the channel vectors of all preceding users, thereby avoiding the interference caused by succeeding users. At the same time, the ZFDPC applies the DPC scheme to eliminate the multiuser-interference caused by preceding users. Mathematically, the ZFDPC can be formulated as a special case of DPC when restricting to one-dimensional signaling for all users, i.e. the case when the covariance matrices $K_k$ are of rank 1 and are decomposed as $K_k = p_k \_{q}_k \_{q}_k^\dag$, for some unit-modular beamforming vector $\_{q}_k$. Consequently, the signal vector of the $k$-th user takes the following form
\beq
\_{x}_k = \sqrt{p_k} x_k \_{q}_k \label{eq:11}
\eeq
for some $x_k \in \C$ satisfying the average power constraint $\E \abs{x_k}^2 \leq 1$. The beamforming vector $\_{q}_k$ is required to be orthogonal to the channel vectors of all preceding users, i.e. it satisfies
\beq
\ell_{i,k} := \_{h}_i^\top \_{q}_k = 0 \quad \tm{ for all $i < k$}. \label{eq:12}
\eeq 
Substituting \eqref{eq:11} and \eqref{eq:12} into the channel input-output equation \eqref{eq:3} shows that the received signal of the $k$-th user is 
\bea
y_k &=& \sum_{i < k} \sqrt{p_i} \_{h}_k^\top \_{q}_i x_i + \sqrt{p_k} \_{h}_k^\top \_{q}_k x_k + w_k \no\\
&=& \sum_{i < k} \sqrt{p_i} \ell_{k,i} x_i + \sqrt{p_k} \ell_{k,k} x_k + w_k. \label{eq:13}
\eea
The first term $\sum_{i < k} \sqrt{p_i} \ell_{k,i} x_i$ of \eqref{eq:13} can be seen as the interference caused by preceding users and can be eliminated by employing DPC  at transmitter in a form similar to \eqref{eq:vk}. Thus, the ZFDPC literally decomposes the MISO BC into a group of parallel interference-free channels, at a cost of certain capacity loss. The values of $\{\_{q}_k\}$ and $\{ \ell_{k,i}\}$ can be easily (and optimally) determined as follows. Let 
\beq
H = \left[
\begin{array}{c}
\_{h}_1^\top\\
\vdots \\
\_{h}_K^\top
\end{array} \right] \label{eq:H}
\eeq
be the overall $(K \x N)$ channel matrix that is completely known to the base-station. Then the optimal choices of $\{\_{q}_k\}$ and $\{ \ell_{k,i}\}$ are given by the QR-decomposition of matrix $H^\dag$;  say $H^\dag = Q L^\dag$, where 
\[
Q = \left[ \_{q}_1 \ \cdots \ \_{q}_K \right]
\]
is an $(N \x K)$ matrix with orthonormal columns, and $L = [\ell_{k,i}]$ is a lower triangular matrix of size $(K \x K)$. The desired signaling vectors $\_{q}_k$ are exactly the column vectors of $Q$ and are proved to be optimal for arbitrary performance measure, whenever the total power constraint is enforced \cite{TuB03, JiangBT06, WieselES08, TranJBO13}. 

To summarize, the ZFDPC  strategy first converts the problem of MISO broadcast communication into that of SISO broadcast communication, by restricting to one-dimensional signaling and $\rank(K_k)=1$, at the cost of a certain loss in the maximal achievable rates.  The SISO broadcast channel is then further ramified --- with help from beamforming and DPC --- into a set of parallel interference-free SISO point-to-point channels, undertaking another loss in capacity, due to the constraints \eqref{eq:12} imposed to eliminate the interference signals from succeeding users. Compared to the DPC (cf. \eqref{eq:9}), the maximal achievable rate of the $k$-th user using ZFDPC is given by 
\beq
R_{\tm{ZFDPC},k} = \log_2 \left( 1 + p_k \abs{\ell_{k,k}}^2 \right), \label{eq:15}
\eeq
provided that the $x_k$'s (cf. \eqref{eq:11}) are encoded by a SISO DPC. Despite a certain capacity loss, it has been shown in \cite{IT_CaireS03} that the ZFDPC remains asymptotically optimal in high and low SNR regimes, provided that the $\{ \_{h}_k\}$ are linearly independent. 

\subsection{Qualitative Approaches of Fairness}  \label{sec:qual} 

It can be seen from \eqref{eq:9} and \eqref{eq:15} that the achievable rates of either DPC or ZFDPC are functions of individual powers $p_1, \ldots, p_K$. Thus, how to distribute the total power $P$ to each user calls for a specific design-objective. Many objectives have been formulated and proposed in the literature, each involving a certain qualitative consideration of fairness. Below we will briefly review four objectives commonly used in the field of wireless network communications.  For simplicity, these objectives will be presented in the form of ZFDPC, and we will write the achievable rate of the $k$-user (cf. \eqref{eq:15}) as $R_{\tm{ZFDPC},k}(\_{p})$, where $\_{p}=[p_1, \ldots, p_K]^\top \in (\R^+)^K$, to emphasize the dependence upon power distribution. These objectives can also be easily rephrased for DPC --- simply replacing $R_{\tm{ZFDPC},k}(\_{p})$ with $R_{\tm{DPC},k}(\_{p})$ (in (cf. \eqref{eq:9}) subject to additional constraints $\tr(K_k) \leq p_k$ for $k=1, \ldots, K$. 

The first design-objective seeks to maximize the overall sum rate and is formulated as 
\beq
\begin{array}{lll}
R_{\tm{ZFDPC,ms}} = & \underset{\_{p}}{\tm{maximize}}  & \sum_{k} R_{\tm{ZFDPC,k}}(\_{p}) \\
& \tm{subject to} & \_{\bf 1}^\top \_{p} \leq P, \ \_{p} \succeq \_{0}.
\end{array} \label{eq:ms}
\eeq
This objective has no concern of fairness among the users; it only aims to maximize the sum of transmission rates of all users, which is an ordinary goal in conventional multi-terminal information theory \cite{book_Kramer08}. 

The second design-objective, termed {\em proportional fairness} \cite{Kelly98, BertsimasFT11}, seeks an  optimal generalized Nash solution to the $K$-player problem \cite{Nash50}, where a power distribution $\_{p}_{\tm{pf}}$ is said to be optimal if replacing $\_{p}_{\tm{pf}}$ with any other $\_{p}$ results in a negative aggregate proportional value of each user's rate, i.e.
\beq
\sum_k \frac{R_{\tm{ZFDPC},k}(\_{p}) - R_{\tm{ZFDPC},k}(\_{p}_{\tm{pf}})}{R_{\tm{ZFDPC},k}(\_{p}_{\tm{pf}})} \leq 0. \label{eq:17}
\eeq
Such design-objective follows directly from the Nash standard for fairness: a transfer of resources (i.e. power in our scenario) among  users is favorable and fair if the sum of the percentage increases of each user's rate is positive. The optimal distribution $\_{p}_{\tm{pf}}$ can be found by solving the following optimization problem \cite{BertsimasFT11}
\beq
\begin{array}{lll}
\_{p}_{\tm{pf}} =  \arg 
&  \underset{\_{p}}{\tm{maximize}} & \sum_{k} \log \left( R_{\tm{ZFDPC,k}}(\_{p})\right)\\
& \tm{subject to} & \_{\bf 1}^\top \_{p} \leq P, \ \_{p} \succeq \_{0}
\end{array} \label{eq:pf}
\eeq
and yields the following sum rate for ZFDPC
\beq
R_{\tm{ZFDPC,pf}} = \sum_k R_{\tm{ZFDPC,k}} ( \_{p}_{\tm{pf}}).
\eeq

The third frequently used design-objective is to maximize the ``average'' of the rates of users. Such ``average'' should not be the usual arithmetic mean, for otherwise the objective would be equivalent to the maximal sum rate given in the first design-objective. Instead, researchers resort to the {\em harmonic mean} \cite{LuoZ08}, and the design-objective is to seek a power distribution that maximizes the harmonic mean of the rates, i.e.
\beq
\begin{array}{lll}
\_{p}_{\tm{hm}} %= [p_{1,\tm{hm}}, \cdots, p_{K,\tm{hm}}]^\top
 = \arg & \underset{\_{p}}{\tm{maximize}} & \left(\sum_{k} \frac{1}{  R_{\tm{ZFDPC,k}}(\_{p})  }\right)^{-1}\\
& \tm{subject to} &  \_{\bf 1}^\top \_{p} \leq P,\ \_{p} \succeq \_{0} 
\end{array}\label{eq:hm}
\eeq
and the resulting sum rate is 
\beq
R_{\tm{ZFDPC,hm}} = \sum_k R_{\tm{ZFDPC,k}} ( \_{p}_{\tm{hm}}).
\eeq

The last commonly used design-objective is termed {\em max-min} criterion introduced by Kalai and Smorodinsky \cite{KalS75}. It aims to maximize the lowest rate among all users, thereby improving the worst-case performance. The criterion can be easily formulated as 
\beq
\begin{array}{lll}
\_{p}_{\tm{mm}} 
 = \arg & \underset{\_{p}}{\tm{maximize}} & \min_k R_{\tm{ZFDPC,k}}(\_{p}) \\
& \tm{subject to} &  \_{\bf 1}^\top \_{p} \leq P, \ \_{p} \succeq \_{0}
\end{array} \label{eq:mm}
\eeq
and the resulting sum rate is 
\beq
R_{\tm{ZFDPC,mm}} = \sum_k R_{\tm{ZFDPC,k}} ( \_{p}_{\tm{mm}}).
\eeq

We remark that all the four optimization problems, \eqref{eq:ms}, \eqref{eq:pf}, \eqref{eq:hm} and \eqref{eq:mm}, can be easily solved by standard convex-optimization techniques. Solution to the last design-objective, i.e. max-min criterion, is particularly simple and has an analytical form. The optimum is achieved when all rates are equal, i.e. 
\beq
R_{\tm{ZFDPC,1}}(\_{p}_{\tm{mm}}) = \cdots = R_{\tm{ZFDPC,K}}(\_{p}_{\tm{mm}}), \label{eq:eqrate}
\eeq
hence the power $p_{\tm{mm},k}$ allocated to the $k$th user is given directly by
\beq
p_{\tm{mm},k}  = \frac{P}{\abs{\ell_{k,k}}^2}  \left( \sum_k \frac{1}{\abs{\ell_{k,k}}^2} \right)^{-1}. \label{eq:49}
\eeq

\subsection{Quantitive Approaches of Fairness} \label{sec:fair}

Each of the design objectives discussed in the previous section has its own take of fairness. Without a common and quantitive measure, it is unlikely to tell which design-objective is better, subject to a constraint on the overall sum rate. The sum rate constraint is particularly important; without it the max-min criterion would surely be the fairest, as all rates are equal (cf. \eqref{eq:eqrate}), but at the same time it has the lowest sum rate among the four. This is of course generally not preferred in wireless communications. 

Generally accepted quantitive measures for fairness include entropy-based index \cite{conf_UchidaK09}, Jain's index \cite{Jain84} and $\alpha$-fairness. For simplicity, here we will focus only on Jain's index, while all upcoming discussions can be easily reformulated in terms of other fairness measures. The Jain's index satisfies all the conditions in \cite{conf_LanKCS10} for being a plausible fairness measure and has been widely used in many areas, including wireless communication. The formal definition of Jain's index is reproduced below. 
\begin{defn}
Given  a set of achievable rates $\{ R_k: k=1, \ldots, K\}$, the corresponding Jain's index for fairness is 
\beq
J(\_{\gamma}) = \frac{\left( \sum_k R_k \right)^2}{K \sum_k R_k^2} = \frac{1}{K} \frac{1}{\norm{\_{\gamma}}_2^2} \label{eq:jain}
\eeq
where 
\beq
\_{\gamma} = \left[ \frac{R_1}{\sum_k R_k}, \cdots,  \frac{R_K}{\sum_k R_k} \right]^\top \label{eq:normal}
\eeq
is the normalized rate vector.  
\qed
\end{defn}

It should be noted that in \eqref{eq:jain} we have reformulated  the Jain's index in terms of the normalized rate vector $\_{\gamma}$, rather than the actual rates $ R_k$'s that generally appear in literatures \cite{Jain84, SediqGSY13, ShiPON14}. The new formulation in $\_{\gamma}$ actually gives a better insight into how the function $J(\_{\gamma})$ measures fairness in general --- much more than simply being the ratio of the squared first moment to the second moment. \alert{To this end, note that we have $\frac1K \leq J(\_{\gamma}) \leq 1$ for any set of rates, where the upper bound represents the case of 100 per cent fairness, and the equality holds if and only if all rates are equal, i.e. when the normalized rate vector equals $\_{e} = \frac{1}{K} \_{\bf 1}$. The lower bound $\frac 1K$ of $J(\_{\gamma})$ comes from the fact that $\norm{\_{\gamma}}_2 \leq \norm{\_{\gamma}}_1 = 1$ and \eqref{eq:jain}. }
The proposition below then shows that the Jain's index $J(\_{\gamma})$ can be related to the geometric angle between vectors $\_{\gamma}$ and $\_{e}$: the smaller the angle is, the more similar $\_{\gamma}$ is to $\_{e}$; hence the fairer the rates are to the users, and the closer to value 1 the Jain's index should be. 

\begin{prop} \label{prop:1}
Given the normalized rate vector $\_{\gamma}$, let $\theta$ be the angle between $\_{\gamma}$ and $\_{e} = \frac{1}{K} \_{\bf 1}$ in the Euclidean inner product vector space $\R^K$. Then the Jain's index can be alternatively defined as 
\beq
J(\_{\gamma}) = \abs{ \cos(\theta)}^2. \label{eq:28}
\eeq
\end{prop}
\begin{proof}
It follows directly from the definition of cosine in $\R^K$, namely,
\[
\abs{ \cos(\theta)}^2 = \frac{\abs{\< \_{\gamma}, \_{e}\>}^2}{\norm{\_{e}}_2^2 \norm{\_{\gamma}}_2^2} =  \frac{1}{K} \frac{1}{\norm{\_{\gamma}}_2^2},
\]
where the last equality is due to $\< \_{\gamma}, \_{e}\>=\frac{1}{K}$ and $\norm{\_{e}}_2^2 = \frac{1}{K}$. 
\end{proof}
\mvs

Unfortunately,  Proposition \ref{prop:1} also reveals a disadvantage of Jain's index when used as the fairness measure, despite satisfying all the axioms listed in \cite{conf_LanKCS10}. To see this, note that from Taylor's expansion of $\cos(\theta)$ in \eqref{eq:28} we have \footnote{For ease of discussion, the standard $\Theta$ notation is introduced here. We say $f(\theta)=\Theta(g(\theta))$ as $\theta \to 0$  if there exist constants $k_1$ and $k_2$, with  $0 < k_1 < k_2 < \infty$, such that $k_1 \abs{g(\theta)} \leq \abs{f(\theta)} \leq k_2 \abs{g(\theta)}$ as $\theta \to 0$. }
\[
J(\_{\gamma}) = 1 - \theta^2 + \Theta(\theta^4)
\]
for small values of $\theta$, i.e. $\abs{\theta} \ll 1$. When $\_{\gamma}$ is close to $\_{e}$ (with respect to any normed distance in $\R^K$), we have $\cos(\theta)$ close to $1$ and $\theta$ close to $0$; then the Jain's index $J(\_{\gamma})$ becomes insensitive to the variation of $\theta$ because of the quadratics. Such phenomenon is particularly pronounced in the high SNR regime, where the total power $P$ is large and all users have similar rates for transmission. In other words, when $P$ is large, all the approaches for qualitative fairness discussed in the previous section, either max sum rate, proportional fairness, harmonic-mean or max-min, have the same form of asymptotic rate $\log_2(P)+O(1)$ for each user. 
Then, the rates of various approaches differ only in the $O(1)$ term, which quickly vanishes when computing the ratios in $\_{\gamma}$. It is therefore difficult to tell which approach yields a better fairness if the Jain's index were used as the quantitive fairness measure. An alternative measure for fairness is thus proposed below to avoid such disadvantage. 
\begin{defn}
Given the normalized rate vector $\_{\gamma}$ defined in \eqref{eq:normal}, the proposed measure for fairness is 
\beq
F(\_{\gamma}) := 1- \frac{K}{2(K-1)}\norm{\_{\gamma}- \_{e}}_1 = 1 - \frac{K}{2(K-1)} \sum_{k=1}^K \abs{\frac1K - \gamma_k},\label{eq:newfair}
\eeq
where $\norm{\_{\gamma}- \_{e}}_1$ is the $\ell_1$ distance from $\_{\gamma}$ to the equal-rate vector $\_{e}$. \qed 
\end{defn}
\svs

Same as the Jain's index, we have $F(\_{\gamma})=1$ if and only if  $\_{\gamma}=\_{e}$ for the 100 per cent fairness. The factor $\frac{K}{2(K-1)}$ appearing in the second term of $F(\_{\gamma})$ comes from a consideration occurred in the most unfair case. Note that the normalized rate vector $\_{\gamma}$ is nonnegative and has a unit $\ell_1$ norm, i.e. $\_{\bf 1}^\top \_{\gamma}=1$. It is easy to show that the largest $\ell_1$ distance for $\_{\gamma}$ to deviate from $\_{e}$ is 
\[
\sup \left\{ \norm{\_{\gamma}- \_{e}}_1 \ : \ \_{\gamma} \succeq \_{0}, \ \norm{\_{\gamma}}_1=1 \right\} = \frac{2(K-1)}{K}.
\]
Thus, with the additional factor $\frac{K}{2(K-1)}$ in \eqref{eq:newfair} we have normalized the range of fairness function $F(\_{\gamma})$ to $[0,1]$, where the minimum value $0$ now indicates the most unfair distribution of rates among the users. 

The function $F(\_{\gamma})$ also satisfies the axioms of continuity, homogeneity, asymptotic saturation, and monotonicity, that are required for being a plausible fairness measure \cite{conf_LanKCS10}, and is much easier to compute than the Jain's index. The resolution range of $F(\_{\gamma})$ --- from 0 to 100 per cent --- is independent of the number of users and is broader than the Jain's index, which has value $\frac 1K$ for the most unfair case. It is also much more sensitive to the small changes of $\theta$ as shown by the following theorem. 
\begin{thm}
Given the normalized rate vector $\_{\gamma}$, we have 
\beq
F(\_{\gamma}) = 1 - \Theta(\tan(\theta)) = 1 - \Theta(\theta)
\eeq
for small $\theta$ defined in Proposition \ref{prop:1}. 
\end{thm}
\begin{proof}
It follows from the proof of Proposition \ref{prop:1} that 
\bean
\norm{\_{e} - \_{\gamma}}_2 &=& \left[ \norm{\_{e}}_2^2 + \norm{\_{\gamma}}_2^2 - 2 \<\_{e}, \_{\gamma}\> \right]^{\frac 12}\\
& =& \left[\frac 1K + \frac{1}{K (\cos(\theta))^2} - \frac2K \right]^{\frac 12} \\
&=& \frac{1}{\sqrt{K}} \abs{\tan(\theta)}.
\eean
The proof is complete after invoking the following standard inequality between $\ell_1$ and $\ell_2$ norms
\[
\norm{\_{e} - \_{\gamma}}_2 \leq \norm{\_{e} - \_{\gamma}}_1 \leq \sqrt{K} \norm{\_{e} - \_{\gamma}}_2
\]
and hence
\begin{multline*}
F(\_{\gamma}) = 1 - \Theta\left( \norm{\_{e} - \_{\gamma}}_1\right) \\ = 1 - \Theta\left( \norm{\_{e} - \_{\gamma}}_2\right) = 1 - \Theta(\tan(\theta)).
\end{multline*}
\end{proof}

It follows that  the newly proposed fairness measure $F(\_{\gamma})$ varies linearly with $\theta$ when $\_{\gamma}$ is in the neighborhood of $\_{e}$ --- a region of far more interest in practice. Hence, it can provide a better resolution for comparing the fairness achieved by various qualitative approaches. 

\section{Tri-stage approach} \label{sec:tri} 

In the previous section we had discussed four commonly used design-objectives for distributing powers among users, namely, the max sum rate, proportional fairness, harmonic mean and max-min criteria, each offering a different sum rate and a different degree of fairness. The max sum rate problem (cf. \eqref{eq:15} and \eqref{eq:ms}), in particular, can be easily solved by water-filling method and results in the largest possible sum rate. However, when the total power $P$ is small, users with small values of $\abs{\ell_{k,k}}^2$ (cf. \eqref{eq:15}) will not be given any power --- because of the water-filling strategy --- and therefore have zero transmission rate. Hence, the max sum rate criterion has unfortunately the worst fairness-performance among the four. The max-min criterion, on the other hand, offers an undisputed fairness as all users have the same transmission rate,  but it has the smallest sum rate. The other two objectives, the proportional fairness and harmonic mean criteria, provide moderate sum rates and reasonable fairnesses. The aim of this section is to come up with a new design-objective that can offer a larger sum rate and  at the same time better fairness, compared to those obtained by the proportional fairness criterion. The proposed method consists of three stages: cake-cutting, mixing and selection stages. For simplicity, it will be presented in the form of ZFDPC, while it can be easily reformulated for other strategies such as DPC to achieve better performance, but at a cost of much higher computational complexity.

\subsection{The First Stage: Cake-cutting}

In the first stage, we divide the overall power $P$ --- namely, the cake --- into two portions, $cP$ and $(1-c)P$ for some $c \in [0,1]$. The portion $cP$ is first distributed among all users following the max-min objective. This would ensure that all users are included in the broadcast and no one is left behind, thereby establishing a basic guarantee of fairness. Once the ``basic'' need of each user is satisfied, we move on to the second wave of power allocation and distribute the remaining power $(1-c)P$ to all users; only this time we aim to maximize the overall sum rate. 

Specifically, given any $c \in [0,1]$, we first distribute the power $cP$ among all users employing the max-min objective for optimization
\beq
\begin{array}{ll}
\tm{maximize} & \min_k R_{\tm{ZFDPC},k}(\_{p}) \\
 \tm{subject to} & \_{\bf 1}^\top \_{p} \leq cP \tm{ and } \_{p} \succeq \_{0}
\end{array} \label{eq:60}
\eeq
By \eqref{eq:49} the optimal solution to the above problem is 
\beq
\_{p}_1(c) = \frac{cP}{\sum_k \frac{1}{\abs{\ell_{k,k}}^2}} \left[ \frac{1}{\abs{\ell_{1,1}}^2}, \ldots, \frac{1}{\abs{\ell_{K,K}}^2}\right]^\top. 
\eeq
The next step is to distribute the remaining power $(1-c)P$ among all users, aiming to maximize the overall sum rate. Hence, we seek solutions to the following optimization problem
\beq
\begin{array}{lll}
\_{p}_2(c) = \arg & \underset{\_{p}}{\tm{maximize}} & \sum_k R_{\tm{ZFDPC},k} (\_{p} + \_{p}_1(c)) \\
&  \tm{subject to} & \_{\bf 1}^\top \_{p} \leq (1-c)P  \tm{ and }  \_{p} \succeq \_{0}
\end{array} \label{eq:62}
\eeq
Finally, the overall power-allocation vector (parameterized by the splitting factor $c$) is the sum of the power-allocation vectors from both steps,
\beq
\_{p}(c) = \_{p}_{\tm{1}}(c) + \_{p}_{\tm{2}}(c), \label{eq:34}
\eeq
and this completes the first stage of the proposed design. 

The power-allocation vector $\_{p}(c)$ in turn gives transmission rate $R_{\tm{ZFDPC},k}(\_{p}(c))$ of user $k$ and  an overall sum rate
\beq
\tmR_{\tm{sum}}(c) := \sum_k R_{\tm{ZFDPC},k}(\_{p}(c)). \label{eq:35}
\eeq
The corresponding fairness value, measured by the newly proposed $\ell_1$-based function in \eqref{eq:newfair}, equals 
\beq
\tmF(c) := F(\_{\gamma}(c)) = 1 - \frac{K}{2(K-1)}\norm{\_{\gamma}(c) - \_{e}}_1 \label{eq:36}
\eeq
where $\_{e}=\frac 1K \_{\bf 1}$ is the equal rate vector and 
\[
\_{\gamma}(c)=\frac{1}{\tmR_{\tm{sum}}(c)}\left[ R_{\tm{ZFDPC},1}(c), \ldots, R_{\tm{ZFDPC},K}(c) \right]^\top. 
\]
Thus, by the end of this stage we have obtained a wide spectrum of achievable sum rates and fairnesses $\{ (\tmR_{\tm{sum}}(c) , \tmF(c)) : c \in [0,1]\}$. Note that both max sum rate and max-min criteria are covered by setting $c$ equal to $0$ and $1$, respectively.

\subsection{The Second Stage: Mixing} 

The key to the second stage is the observation that wireless channels are often quasi-static, meaning that the channel vectors $\_{h}_k$ can hold constant and unchanged for many channel uses. Within the constant-valued channel block, a fixed choice of $c$ is not necessarily the best for all transmissions, and a mixed strategy of using various $c$'s could actually offer better sum rate and fairness. Finding the best mixed strategy calls for an optimal tradeoff between achievable sum rate and fairness. Given the desired sum rate $\tmR$, the best fairness that our design can offer is characterized by the following function
\begin{multline}
\tmF_{\max}(\tmR) := \sup \left\{ \int_0^1 \tmF(c) w(c) \, d c \ : \right.\\
 \int_0^1 \tmR_\tm{sum}(c) w(c) \, dc = \tmR, \ \int_0^1 w(c)\, dc = 1, \\
\left. \tm{ for all functions $w(c) \in [0,1]$}\right\}, \label{eq:67}
\end{multline}
\alert{which is the outer convex hull of the achievable sum rate-fairness pairs $\{ (\tmR_{\tm{sum}}(c) , \tmF(c)) : c \in [0,1]\}$.} The function $w(c)$ plays the role of time-sharing among all possible ways to cut a cake: the probability of $c$ being used for the cake-cutting is exactly $w(c)$. In other words, for a quasi-static block of $T$ channel uses, the vector $\_{p}(c)$ will be used for power allocation in approximately $T\cdot w(c)$ times. Equation \eqref{eq:67} then seeks the best possible strategy for combining the $c$'s to yield the largest possible degree of fairness, provided that the average sum rate $\tmR$ is to be achieved. 

\subsection{The Third Stage: Selection} \label{sec:sel}

\begin{figure}[t!] 
\[
\includegraphics[width=3in]{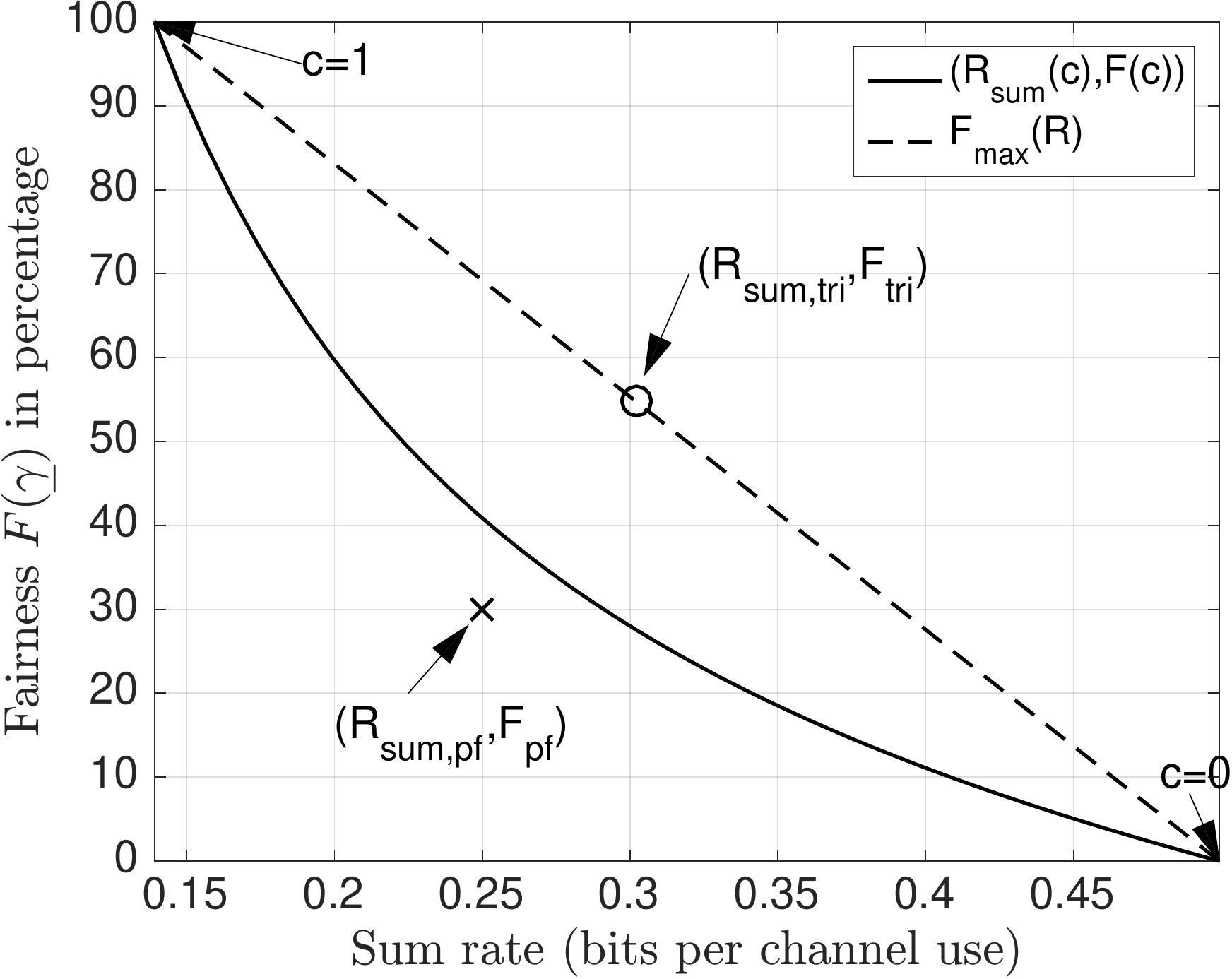}
\]
\caption{An illustration of the proposed selection method.} \label{fig:1}
\end{figure}

The optimal sum rate-fairness tradeoff $\tmF_{\max}(\tmR)$ given in \eqref{eq:67} provides a set of sum rate-fairness pairs $\{ (\tmR,\tmF_{\max}(\tmR)): \tmR \in \R \}$ that can be achieved by our ``mixed'' power-allocation strategy. It does not say which pair should be chosen for operation. In fact, it is impossible to make any rational choice of the operation point, without having an explicit rule dictating quantitatively how much more important the sum rate to the fairness is. The Nash standard \cite{Nash50,BertsimasFT11} argues --- in the qualitative sense of fairness --- that the best sum rate-fairness pair should satisfy that any deviation from the power distribution $\_{p}_{\tm{pf}}$ (cf. \eqref{eq:pf}) would incur a negative aggregate proportional change of the rates (cf. \eqref{eq:17}). The standard does not invoke explicitly (or implicitly) any quantitive measure of fairness; therefore it is not necessarily the best selection rule from the viewpoint of 
 tradeoff $\tmF_{\max}(\tmR)$. Specifically, the proportional fairness criterion yields the following sum rate 
\beq
\tmR_{\tm{sum,pf}} := \sum_k R_{\tm{ZFDPC},k}(\_{p}_{\tm{pf}})
\eeq
and fairness 
\beq
\tmF_{\tm{pf}} := F(\_{\gamma}_{\tm{pf}})
\eeq
where 
\[
\_{\gamma}_{\tm{pf}} = \frac{1}{\tmR_{\tm{sum,pf}}}\left[ R_{\tm{ZFDPC},1}(\_{p}_{\tm{pf}}), \ldots, R_{\tm{ZFDPC},K}(\_{p}_{\tm{pf}}) \right]^\top.
\]
It is easy to show that for $K=2$ 
\beq
\tmF_{\max}(\tmR_{\tm{sum,pf}}) \geq \tmF_{\tm{pf}}, \label{eq:40}
\eeq
meaning that the proportional fairness criterion cannot be any fairer than the tri-stage design, \alert{with respect to the newly proposed fairness measure \eqref{eq:newfair}.} In other words, as far as the two-player problem \cite{Nash50} is concerned, the proposed tri-stage scheme could potentially offer performance better than proportional fairness, which is derived from the Nash standard \cite{BertsimasFT11}. There are at least three implications from the inequality \eqref{eq:40}: 1) we could use the proposed tri-stage method to send information to all users at rate $\tmR_{\tm{sum,pf}}$ --- same as that of proportional fairness --- and obtain a better fairness $\tmF_{\max}(\tmR_{\tm{sum,pf}})$, or 2) we could send at a  rate slighter higher than $\tmR_{\tm{sum,pf}}$ and get a fairness value between $\tmF_{\tm{pf}}$ and $\tmF_{\max}(\tmR_{\tm{sum,pf}}) $, or 3) we could sacrifice a little bit on the rate in exchange for a better fairness. As already said, all is not clear unless we are given a specific rule weighing the importances between sum rate and fairness. On the other hand, when such rule is not available, the proportional fairness criterion remains to be  a good starting point to search for operation points on $\tmF_{\max}(\tmR)$ as the Nash standard has performed very well in many areas of engineering and economics since its first appearance. Thus, given the pair  $(\tmR_{\tm{sum,pf}},\tmF_{\tm{pf}})$ from  the proportional fairness criterion,  we propose to seek a point on curve $\tmF_{\max}(\tmR)$ that offers slightly better performance than 
$(\tmR_{\tm{sum,pf}},\tmF_{\tm{pf}})$ --- a small wish that can be easily granted. The exact proposed operating pair is the following 
\begin{multline}
(\tmR_{\tm{sum,tri}}, \tmF_{\tm{tri}}) := \\ \left\{
\begin{array}{l}
\arg \min_{(\tmR,\tmF_{\max}(\tmR))}  \left[ \abs{\tmR-\tmR_{\tm{sum,pf}}}^2 + \abs{\tmF_{\max}(\tmR) - \tmF_{\tm{pf}}}^2 \right],\\
 \hfill \tm{ if $\tmF_{\max}(\tmR_{\tm{sum,pf}}) \geq \tmF_{\tm{pf}}$},\\\\
(\tmR_{\tm{sum,pf}}, \tmF_{\tm{pf}}), 
\hfill \tm{ if otherwise},
\end{array} \right. \label{eq:73}
\end{multline}
where the second case is just an insurance policy in case our wish is denied, \alert{and where the first case echos the third implication discussed in the above paragraph.} A simple illustration of the proposed selection method is given Figure \ref{fig:1}, where the pairs $(\tmR_{\tm{sum}}(c),\tmF(c))$ are from the cake-cutting stage for some $c \in [0,1]$. It should be noted that the solid line formed by $\{(\tmR_{\tm{sum}}(c),\tmF(c)): c \in [0,1]\}$ is in general not convex; we only make it so for the illustration purpose. The dashed line represents the optimal sum rate-fairness tradeoff $\tmF_{\max}(\tmR)$ resulting from the mixing stage. The point marked by an X is  $(\tmR_{\tm{sum,pf}}, \tmF_{\tm{pf}} )$ obtained by the proportional fairness criterion; the  point marked by an O is the proposed operating point $(\tmR_{\tm{sum,tri}}, \tmF_{\tm{tri}})$ given by the selection rule \eqref{eq:73} and can be achieved by using a statistical power allocation method outlined in Algorithm \ref{alg:1}. 

\begin{algorithm}[!t]
\caption{The Proposed Statistical Power Allocation Strategy to Achieve $(\tmR_{\tm{sum,tri}}, \tmF_{\tm{tri}})$} \label{alg:1}
\begin{algorithmic}[1]
\REQUIRE The desired average sum rate $\tmR_{\tm{sum,tri}}$ and desired average fairness $\tmF_{\tm{tri}}$ with $\tmF_{\max}(\tmR_{\tm{sum,tri}})=\tmF_{\tm{tri}})$
\STATE Find an optimal solution $w^*(c)$ to the optimization problem \eqref{eq:67}, i.e. the nonnegative-valued function $w^*(c)$ satisfies
\bean
&& \int_0^1 \tmF(c) w^*(c) \, d c = \tmF_{\tm{tri}}, \\
&& \int_0^1 \tmR_\tm{sum}(c) w^*(c) \, dc = \tmR_{\tm{sum,tri}}, \ \tm{ and } \\
&& \int_0^1 w^*(c)\, dc = 1,
\eean
where $\tmR_\tm{sum}(c)$ and $\tmF(c) $ are defined in \eqref{eq:35} and \eqref{eq:36}, respectively. 
\STATE Generate a random value $C \in [0,1]$ based on the  probability density function $w^*(c)$. Use the power allocation vector $\_{p}(C)$ (cf. \eqref{eq:34}) obtained from the first stage to distribute the total power $P$ among users and perform broadcast communication using ZFDPC. 
\STATE Repeat Step 2 for various $C$'s within the quasi-static channel block. 
\end{algorithmic}
\end{algorithm}

To summarize, the proposed tri-stage design first finds a tradeoff $\tmF_{\max}(\tmR)$ between the sum rate and fairness based on a {\em statistical} two-step cake-cutting method --- the parameter $c$ from the first stage is seen as a random variable in the second stage. We then use the proportional fairness as a starting point to choose the operating point $(\tmR_{\tm{sum,tri}}, \tmF_{\tm{tri}})$ on $\tmF_{\max}(\tmR)$. \alert{For quasi-static channels,} the proposed design also offers a statistical power-allocation method for MISO BC using ZFDPC (or any other communication schemes). This is in sharp contrast to all current communication schemes in which only a single, deterministic power-allocation is used. It will be seen from simulation results in Section \ref{sec:sim} that the proposed statistical power-allocation could offer not only larger sum rate but also better fairness than the commonly used design objectives.

\section{Further Remarks and Performance Results}%%%20181016

\subsection{Upper Bound on Rate-Fairness Tradeoff: Rate-Split} 

Given the channel vectors $\{ \_{h}_k\}$ of all users (or equivalently the channel matrix $H$ in \eqref{eq:H}) we have shown that the base-station can employ either DPC, ZFDPC, or other coding schemes to simultaneously send information to all $K$ users under a total power constraint. The proposed tri-stage  design-objective, as outlined in Section \ref{sec:tri}, provides an efficient strategy to distribute power among all users, taking both sum rate and fairness into consideration. When the channel is quasi-static, we further showed a mixed strategy can provide performance, in terms of both sum rate and fairness, much better than the commonly used design objectives. The mixed strategy also gives an optimal tradeoff --- in terms of the tri-stage design --- between the sum rate and fairness. The tradeoff is explicitly characterized by the function $\tmF_{\max}(\tmR)$, and has an implicit dependence upon the channel matrix $H$. In this section we will focus on such dependence and will write the tradeoff as $\tmF_{\max}(H,\tmR)$  for emphasis. 

Given the channel matrix $H$, the tradeoff $\tmF_{\max}(H,\tmR)$ is optimal when the tri-stage design-objective is employed and when the coding is restricted to take place within the channel block where $H$ is fixed. In other words, for a series of channel blocks $H_1, H_2, \ldots, H_L$, the proposed tri-stage approach gives the sum rate-fairness pair $(\tmR_{\tm{sum,tri}}(H_\ell),\tmF_{\tm{tri}}(H_\ell))$ at the $\ell$-th channel block $H_\ell$, and achieves the following asymptotic average sum rate and fairness for an ergodic channel
\bea
\bar{\tmR}_{\tm{sum,tri}} &:=& \lim_{L \to \infty} \frac{1}{L} \sum_{\ell=1}^L \tmR_{\tm{sum,tri}}(H_\ell) = \E_H \tmR_{\tm{sum,tri}}(H), \label{eq:42}\\
\bar{\tmF}_{\tm{tri}} &=& \lim_{L \to \infty} \frac{1}{L} \sum_{\ell=1}^L \tmF_{\tm{tri}}(H_\ell) = \E_H \tmF_{\tm{tri}}(H). \label{eq:43}
\eea
%for any ergodic channel and  for each channel realization $H$, assuming that for example in the DPC case the outage probability is asymptotically zero as $P \to \infty$, i.e. 
%\[
%\Pr \left\{ \sum_k \log_2 \left( 1+\frac{ \_{h}_{\Pi(k)}^\dag K_k \_{h}_{\Pi(k)}}{1+\_{h}_{\Pi(k)}^\dag \left( \sum_{i > k} K_i \right) \_{h}_{\Pi(k)}} \right) \geq R \right\} \doteq 1
%\]
On the other hand, if  coding is allowed to take place across multiple channel blocks, then we can have different sum rate for each block, as long as their asymptotic average still equals $\bar{\tmR}_{\tm{sum,tri}}$. This could potentially lead to a larger value of fairness. Specifically, given a series of $L$ channel realizations, $H_1, H_2, \ldots, H_L$, let 
\beq
R_{\tm{sp}}^{(L)} : (H_1, \ldots, H_L) \mapsto (R_{\tm{sum,1}},\ldots,R_{\tm{sum,L}})
\eeq
be a rate-split function that assigns sum rate $R_{\tm{sum},\ell}$ to the $\ell$-th channel block $H_\ell$. With the tri-stage approach and rate-splitting strategy $R_{\tm{sp}}^{(L)}$, it achieves average rate $\frac{1}{L} \sum_\ell R_{\tm{sum},\ell}$ and average  fairness $\frac{1}{L} \sum_\ell \tmF_{\max}(H_\ell, R_{\tm{sum},\ell})$.  Therefore, for any desired average sum rate $\bar{R}_{\tm{sum}}$ we could optimize over all possible rate-split functions $R_{\tm{sp}}^{(L)}$  to obtain a better fairness. In other words, if the mixing stage of  tri-stage design is allowed to take place in multiple channel blocks of an ergodic quasi-static channel, then the optimal sum rate-fairness tradeoff is given by
\begin{multline}
\bar{\tmF}_{\max}^\star(\bar{\tmR}_{\tm{sum}}) = \limsup_{L \to \infty} \biggl\{ \frac{1}{L} \sum_{\ell=1}^L {\tmF}_{\max}(H_\ell, R_{\tm{sum},\ell}) \ : \\
R_{\tm{sp}}^{(L)} (H_1, \ldots, H_L)= (R_{\tm{sum,1}},\ldots,R_{\tm{sum,L}}),\\
R_{\tm{sum,1}}+\cdots +R_{\tm{sum,L}} = L \bar{\tmR}_{\tm{sum}} \biggr\}
\label{eq:75}
\end{multline}
where the supremum is taken over all possible rate-split functions $ R_{\tm{sp}}^{(L)}$ having average sum rate $\bar{\tmR}_{\tm{sum}}$. We remark that the function $\bar{\tmF}_{\max}^\star(\bar{\tmR}_{\tm{sum}}) $ can be easily evaluated using Monte-Carlo methods. 

This ultimate tradeoff $\bar{\tmF}_{\max}^\star(\bar{\tmR}_{\tm{sum}})$, though has performance superior to the single-block achievable pair $(\bar{\tmR}_{\tm{sum,tri}}, \bar{\tmF}_{\tm{tri}})$, is in fact unachievable in practice. Achieving $\bar{\tmF}_{\max}^\star(\bar{\tmR}_{\tm{sum}})$ requires the non-causal knowledge of $H_1, H_2, \ldots,$ at the base-station, hence we shall regard $\bar{\tmF}_{\max}^\star(\bar{\tmR}_{\tm{sum}})$ as the upper bound for the rate-fairness tradeoff. 

The ultimate tradeoff $\bar{\tmF}_{\max}^\star(\bar{\tmR}_{\tm{sum}})$ is also an excellent performance benchmark for the selections of operating sum rate-fairness pairs from the single-block tradeoff $\tmF_{\max}(\tmR)$. Recall that in the last stage of the proposed tri-stage approach we were asked to decide which operating point from the single-block tradeoff $\tmF_{\max}(\tmR)$ should be chosen for transmission, and we made our choice based on the reputable proportional fairness --- choosing the point on $\tmF_{\max}(\tmR)$ that is the closest to the point of proportional fairness. Our intuitive choice achieves an average sum rate $\bar{\tmR}_{\tm{sum,tri}}$ (cf. \eqref{eq:42}) and an average fairness $\bar{\tmF}_{\tm{tri}} $ (cf. \eqref{eq:43}) in the long run. The ultimate tradeoff \eqref{eq:75} then helps us to see how far off  the performance of our choice to that of the best possible selection (if there is one) is. The tradeoff $\bar{\tmF}_{\max}^\star(\bar{\tmR}_{\tm{sum}})$ also serves as the ultimate benchmark for all possible selection schemes that can be considered and/or be eventually employed in the last stage of the proposed design objective. 

\subsection{Performance Results} \label{sec:sim}

\begin{figure}[t!] 
\[
\begin{array}{c}
\includegraphics[width=\columnwidth]{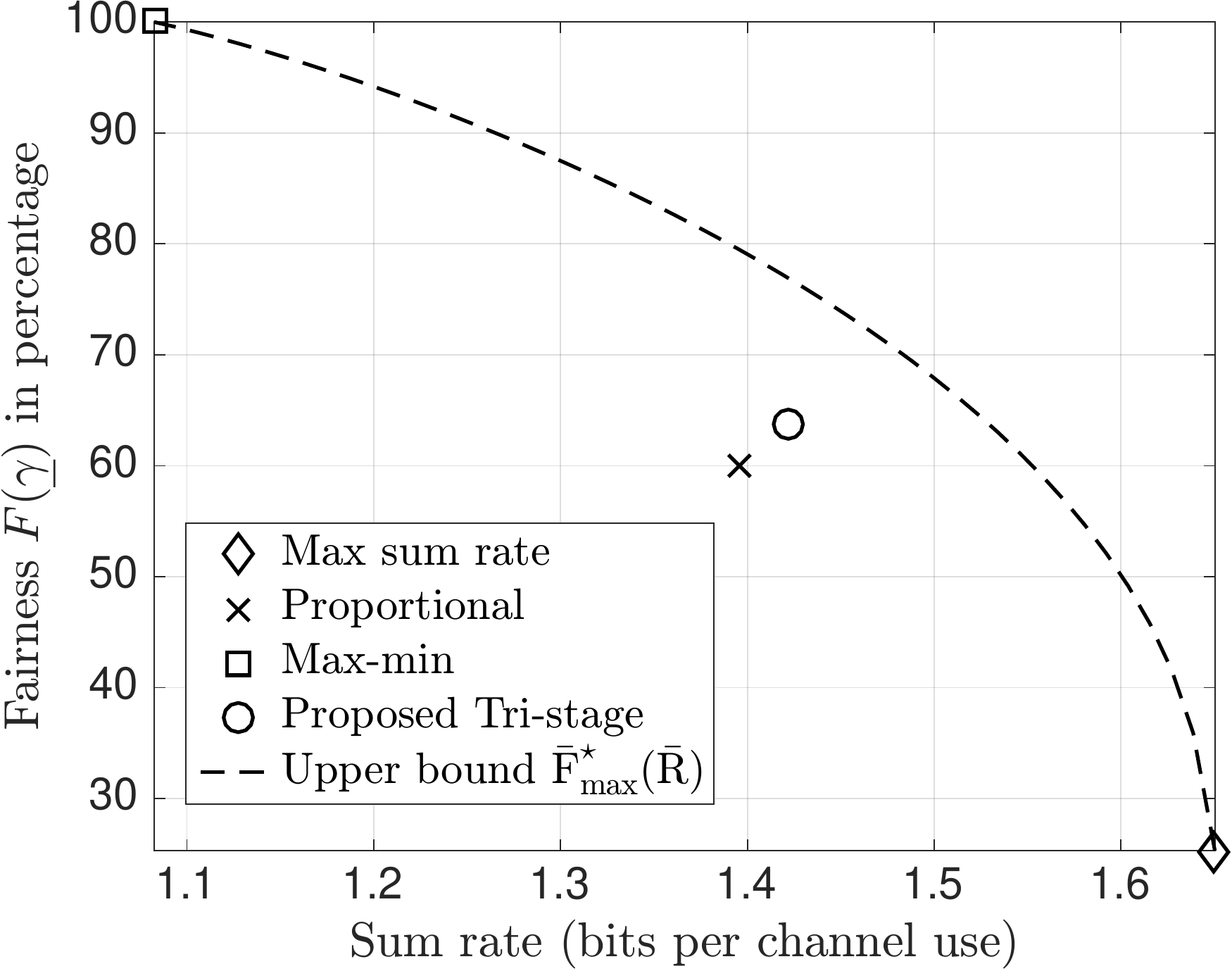} \\
\tm{(a)} \\
\includegraphics[width=\columnwidth]{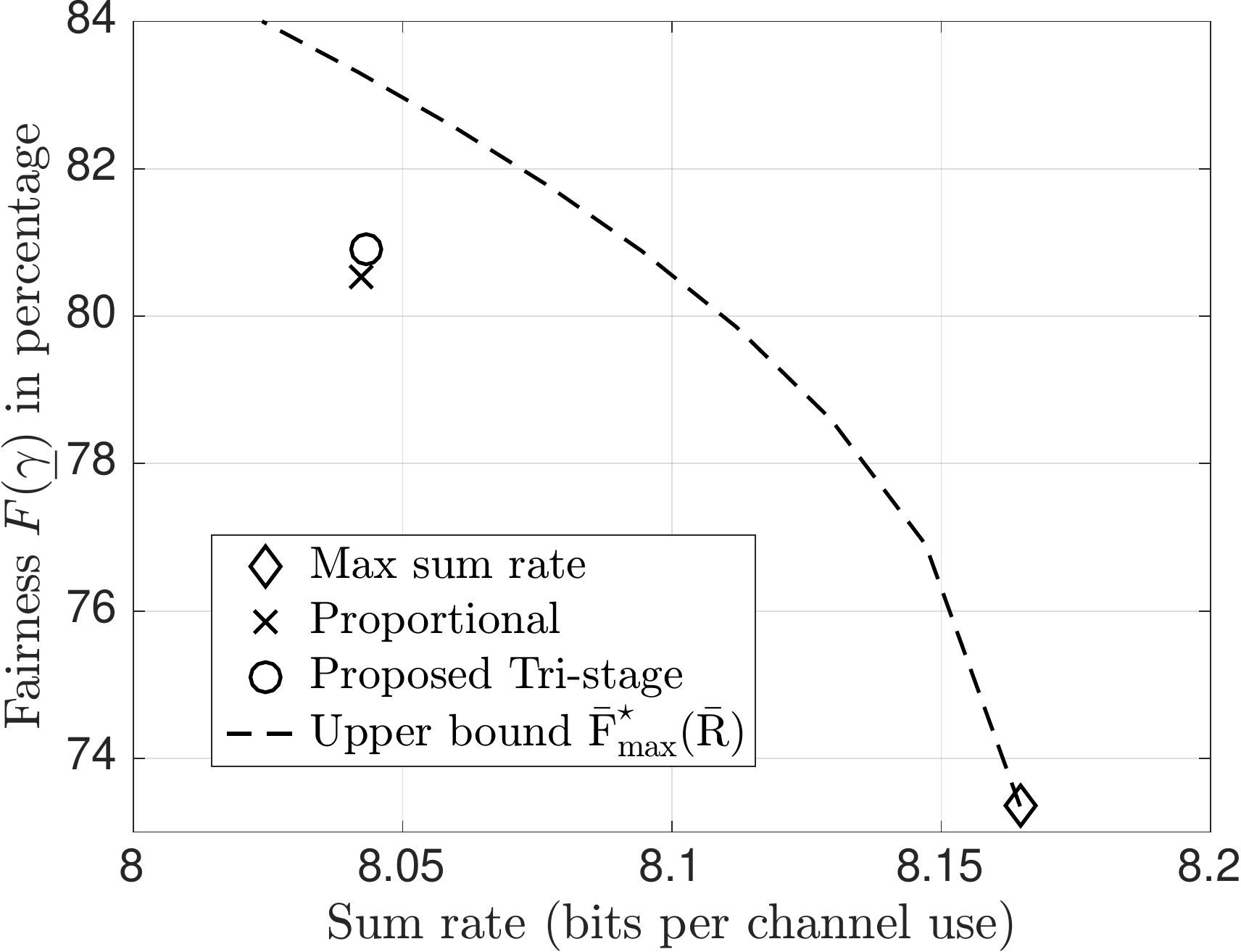}\\
\tm{(b)}
\end{array}
\]
\caption{Achieved average sum rates and fairness values of ZFDPC using max sum rate, proportional fairness, max-min and the proposed tri-stage criteria when $K=2$ users, $N=2$ transmit antennas, and (a) $P=0$ dB and (b) $P=15$dB.} \label{fig:2}
\end{figure}

Two simulation results are presented to demonstrate the performance proposed tri-stage design. While the design is general and can be integrated into any coding scheme used for wireless network communication, here the simulations are performed for MISO BC using 
ZFDPC, due to its simplicity for implementation. In particular, we first consider the case of $K=2$ single-antenna users and $N=2$ transmit antennas at the base-station. The results are given in Figure \ref{fig:2} for (a) $P=0$ dB and (b)$P=15$ dB, representing the performance in low and moderate SNR regimes. In both cases, the sum rate-fairness pairs obtained from the max sum rate criterion always achieve the largest sum rate but at the cost of low fairness value. The level of unfairness is much more pronounced when the total power $P$ is small, as one of the two users could be given zero power and be excluded from communication. When $P$ is large, the encoder has more resources for power distribution, so both users are included in the transmission, and an acceptable --- but not satisfactory --- level of fairness is achieved. The max-min criterion, on the other hand, aims to provide an equal transmission rate to both users,  hence the fairness value is always at 100 percent, but at a cost of certain loss in sum rate. The sum rate-fairness pair derived from the proportional fairness has sum rate smaller than that of max sum rate but much fairer to the users, especially in the low SNR regime.  In both cases, it is seen that the proposed tri-stage criterion has better performance than the proportional fairness, not only a larger sum rate but also better fairness. \alert{The performance gains, however, come at a price of higher computational complexity, since the proposed tri-stage design involves an optimization in the mixing stage to find the best time-sharing among all possible cake cuttings.}

All the achievable sum rate-fairness pairs are upper bounded by the ultimate tradeoff $\bar{\tmF}_{\max}^\star(\bar{\tmR}_{\tm{sum}})$, as achieving the latter calls for the noncausal knowledge of all future channel realizations. It is also seen from Figure \ref{fig:2}.(b) that in moderate SNR regime both proportional fairness and tri-stage designs are fairly close to the ultimate bound $\bar{\tmF}_{\max}^\star(\bar{\tmR}_{\tm{sum}})$, leaving little room for further improvement. But for low SNR regime as in Figure \ref{fig:2}.(a),  it appears possible to improve the rate-selection method in the last stage to yield a better sum rate-fairness pair. 

\alert{A similar simulation for a much more complex case of $K=8$ single-antenna users and $N=8$ transmit antennas transmitting at $P=0$ dB  at base-station is given in Figure \ref{fig:4}. It can be seen that the proposed tri-stage method not only  is much fairer than the proportional fairness but also provides a slightly higher sum rate. }

\begin{figure}[t!] 
\[
\includegraphics[width=\columnwidth]{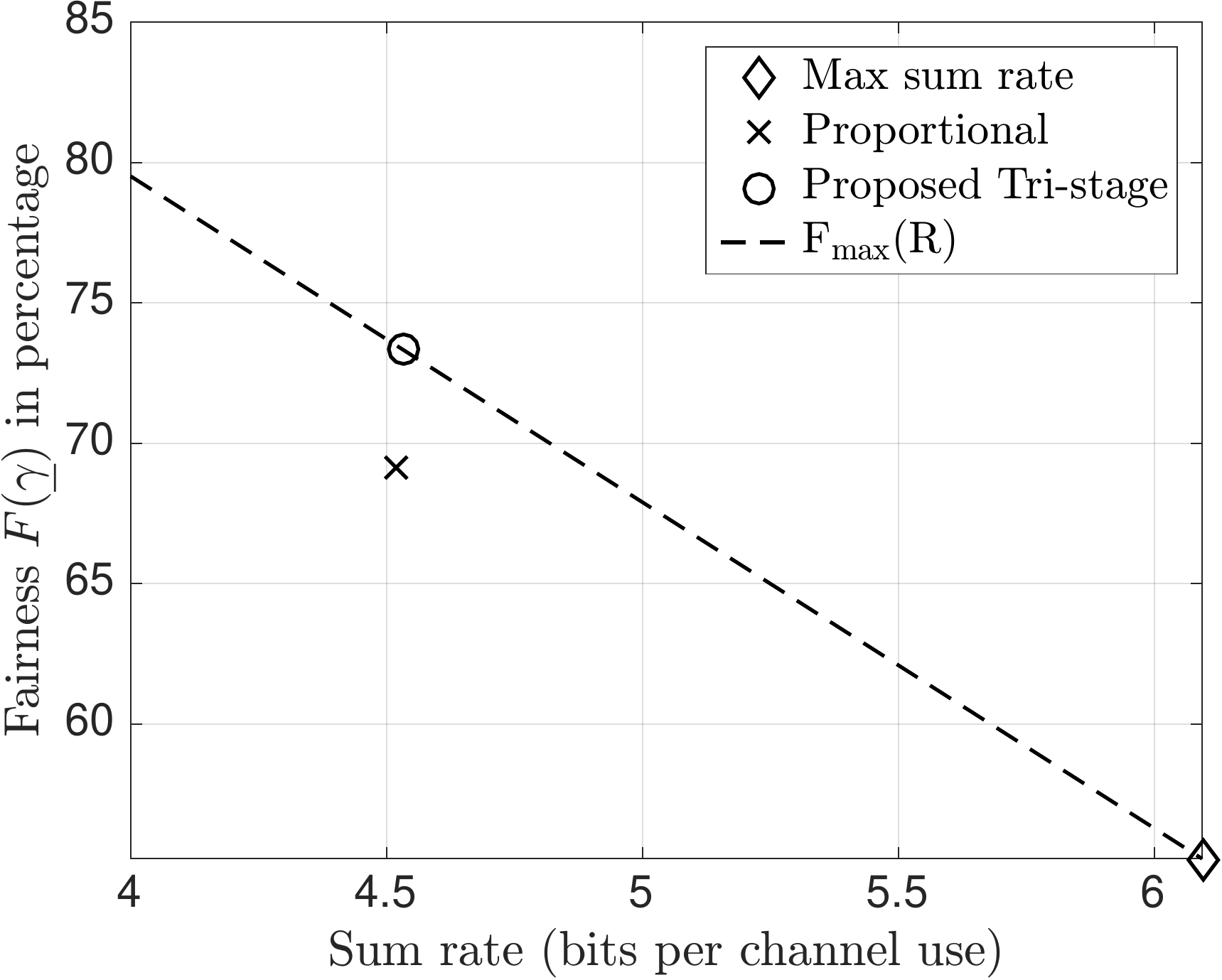} 
\]
\caption{Achieved average sum rates and fairness values of ZFDPC using max sum rate, proportional fairness and the proposed tri-stage criteria when there are $K=8$ users and $N=8$ transmit antennas at $P=0$ dB.} \label{fig:4}
\end{figure}

\section{Conclusion} \label{sec:con}

In this paper we first investigated the tradeoff between sum rate and fairness for MISO BC. The achievable sum rates were based on DPC or ZFDPC, subject to an individual power constraint. Several qualitative approaches for fairness, such as max sum rate, proportional fairness, harmonic mean and max-min, were also discussed. For the quantitive approaches, we showed that the widely used Jain's index could become insensitive at high SNR regime and hence proposed an $\ell_1$-norm based fairness measure that can compare the fairness levels achieved by various design objectives at a much finer resolution. We also introduced a new tri-stage design objective that is based on a new concept of statistical power allocation, in sharp contrast to the fixed, deterministic method used in all existing wireless/wired communication systems. The new scheme randomly allocates powers to users based on an optimal probability distribution derived from the tradeoff between sum rate and fairness. \alert{We also remark that the proposed tri-stage design objective can be easily extended to MIMO BC with successive zero forcing DPC \cite{DabbaghL07} as well as zero-forcing beamforming methods.} Simulation results showed that the proposed approach can simultaneously achieve a larger sum rate and better fairness than the reputable proportional fairness. A performance upper bound was also given in the paper to show that there might still be rooms for further improvement, especially in the low SNR regime. \alert{Finally, it is worth to note that the ordering of users also has some impact on fairness and sum rate. For traditional DPC, all user-orderings have the same sum rate but possibly different fairness values. For ZFDPC, both sum rate and fairness value could change as the ordering of users varies. How to use the ordering of users to improve sum rate and fairness still calls for further research.}


\end{document}